\documentclass[conference]{IEEEtran}\IEEEoverridecommandlockouts

\usepackage{amsmath,amssymb} \interdisplaylinepenalty=2500
\usepackage{cite}
\usepackage{color}
\usepackage{subfigure}
\usepackage{epsfig}
\usepackage{graphicx}

{\begin{itemize}%
\setlength{\itemsep}{0pt}%
\setlength{\topsep}{0pt}%
\setlength{\partopsep}{0pt}%
\setlength{\parskip}{0pt}}%
{\end{itemize}}

\newtheorem{theorem}{Theorem}

\newtheorem{lemma}[theorem]{Lemma}

\def\qed{\endIEEEproof}

\newcommand{\junk}[1]{}
\newcommand{\suppress}[1]{}
\newcommand{\ignore}[1]{}

\newcommand{\poly}{\mbox{poly}}

\newcommand{\Cp}{C}

\newcommand{\eY}{Y}

\newcommand{\fs}{q}

\newcommand{\Field}{\mathbb{F}}

\newcommand{\Block}{n}
\newcommand{\Bl}{n}

\newcommand{\Rate}{R}

\newcommand{\Numzo}{Z_O}
\newcommand{\Numzi}{Z_I}

\newcommand{\eX}{X}

\DeclareMathOperator{\rank}{\sf rank\hspace{0.1em}}

\newcommand{\Fq}{\mathbb{F}_q}

\newcommand{\FQ}{\mathbb{F}_Q}

\newcommand{\mat}[1]{\begin{bmatrix} #1 \end{bmatrix}}

\newcommand{\calS}{\mathcal{S}}

\title{Network Codes Resilient to Jamming and Eavesdropping}

\author{
\begin{tabular}{cccc}
Hongyi Yao & Danilo Silva & Sidharth Jaggi & Michael Langberg
\\
{\small Tsinghua University}&
{\small State University of Campinas}&
{\small Chinese University of Hong Kong}
&{\small The Open University of Israel}
\end{tabular}
\thanks{The work of Hongyi Yao was supported in part by National Natural Science Foundation of China Grant 60553001, the National Basic Research Program
of China Grant 2007CB807900 and 2007CB807901. The work of Danilo Silva was supported by FAPESP grant 2009/15771-7.
The work of Sidharth Jaggi was supported  RGC GRF grant 412608, 411008, and 411209, RGC AoE grant on Institute of
Network Coding, established under the University Grant Committee of Hong Kong, CUHK MoE-Microsoft Key Laboratory of Humancentric
Computing and Interface Technologies, Direct Grant (Project Number 2050397) of The Chinese University of Hong Kong, and two
gift grants from Microsoft and Cisco. The work of Michael Langberg was supported in part by ISF grant 480/08. }
}

\begin{document}
\maketitle

\begin{abstract}
We consider the problem of communicating information over a network secretly and reliably in the presence of a hidden adversary who can eavesdrop and inject malicious errors.
We provide polynomial-time, rate-optimal distributed network codes for this scenario, improving on the rates achievable in~\cite{Ngai.Yeung2009:Secure-Error-Correcting}.
Our main contribution shows that as long as the
sum of the adversary's jamming rate $\Numzo$ and his eavesdropping
rate $\Numzi$ is less than the network capacity $\Cp$, ({\it i.e.,} $\Numzo +
\Numzi < \Cp$), our codes can communicate (with vanishingly small error probability) a single bit correctly and without leaking any information to the adversary. We then use this to design codes that allow communication at the optimal source rate of $\Cp -
\Numzo-\Numzi$, while keeping the communicated message secret from the adversary.
Interior nodes are oblivious to the presence of adversaries and
perform random linear network coding; only the source and
destination need to be tweaked. In proving our results we correct an error in prior work~\cite{Jaggi.Langberg2007} by a subset of the authors in this
work.
\end{abstract}

\section{Introduction}
\label{sec:intro}


A source Alice wishes to transmit information to a receiver Bob over
a network containing a malicious adversary Calvin. Such
scenarios face at least two challenges -- Calvin
might eavesdrop on private communications, or he might disrupt communications by
injecting fake information into the network. In the network coding model
this second danger may be even more pronounced since all nodes,
including honest ones, mix information. In this case, even a small
number of fake packets injected by Calvin may end up corrupting {\it
all} the information flowing in the network, causing decoding
errors.

In this work we consider the {\it secrecy} and {\it {error control}}
issues together. Namely, we design schemes that allow reliable
network communications in the presence of an adversary that can both
jam and eavesdrop, without leaking information to him. In particular,
suppose the network's  min-cut from Alice to Bob
is $\Cp$, and Calvin eavesdrops on
$\Numzi$ links and corrupts $\Numzo$ links\footnote{We consider a
model where network links rather than nodes are eavesdropped and
corrupted; eavesdropping on a node is
equivalent to eavesdropping on links incoming to it, and corrupting
a node is equivalent to corrupting the links outgoing from it.}. We
demonstrate schemes that are distributed, computationally efficient
to design and implement, {and can be used to communicate a {\it single} bit secretly and without error. We then use this scheme as a tool to improve on prior work~\cite{Jaggi++2008}, and achieve a provably optimal rate of $\Cp-\Numzo-\Numzi$.}


Related problems have been considered in the past. Prior results may be classified in the following three categories.

For networks containing adversaries that only eavesdrop on some
links (without jamming transmissions), the work
of~\cite{Cai.Yeung2002:Secure} provided a tight
information-theoretic characterization of the {\it secrecy
capacity}, {\it i.e.}, the optimal rate achievable without leaking
any of Alice's information to Calvin. Efficient schemes achieving
this performance were proposed
by~\cite{Feldman++2004:CapacitySecure,Rouayheb.Soljanin2007,Silva.Kschischang2008:ISIT}.
Cryptographically (but not
information-theoretically) secret schemes for this scenario were
also considered in~\cite{SecureKeyShare}.

For networks containing adversaries with unlimited eavesdropping
capabilities and limited jamming capabilities, prior related work
has focused primarily on the detection of Byzantine
errors~\cite{HoLKMEK:04}, non-constructive bounds on the achievable
{\it {zero-error}} rates~\cite{YeuC:06a,CaiY:06a}, and network
error-correcting codes~\cite{Matsumoto} (which have high design
complexity)
and~\cite{Jaggi++2008,Silva++2008,Kotter.Kschischang2008,Jaggi.Langberg2007}
(which have low design complexity). Results for this setting are
also available under cryptographic
assumptions~\cite{SignLinearSpace1,SignLinearSpace2}.

The scenario closest to the one considered in this work, with
limitations on both Calvin's  eavesdropping power $\Numzi$ and his
jamming power $\Numzo$, have been considered
in~\cite{Jaggi++2008,Jaggi.Langberg2007,Silva2009(PhD),Ngai.Yang2007:DeterministicSEC,Ngai.Yeung2009:Secure-Error-Correcting}.
Under the requirement of \emph{zero} error probability, the maximum rate of secret and reliable communication is given by $\Cp - 2\Numzo - \Numzi$. Schemes achieving this rate have been proposed in~\cite{Ngai.Yang2007:DeterministicSEC,Ngai.Yeung2009:Secure-Error-Correcting}
(high design complexity schemes) and~\cite{Silva.Kschischang2008:Security-IT,Silva2009(PhD),Silva.Kschischang2010:ISIT} (low
design complexity schemes). The optimality of such a rate has been shown in~\cite{Ngai.Yeung2009:Secure-Error-Correcting} for single-letter coding and in~\cite{Silva.Kschischang2010:ISIT} for block coding.

If the requirement of zero error probability is relaxed to \emph{vanishingly small} error probability, as considered here, then higher rates may be achieved. In particular, the work in~\cite{Jaggi++2008} provided computationally efficient communication schemes (but with no guarantees on secrecy) at rate $\Cp - \Numzo$ as long as the technical requirement $\Cp > 2\Numzo + \Numzi$ was satisfied. Work by a subset of the authors of this paper claimed in~\cite{Jaggi.Langberg2007} to improve this technical requirement to $\Cp > \Numzo + \Numzi$. As we demonstrate in Section~\ref{sec:errata}, prior proof of the claim was incorrect, and
Section~\ref{sec:main} gives a correct proof of the claim. Combining these results with the secrecy scheme of \cite{Silva.Kschischang2008:ISIT} allows us to obtain the optimal rate of $\Cp - \Numzo - \Numzi$ when secrecy constraints are incorporated.



\section{Main Results}
\label{sec:main}
The main results of this work are
Theorems~\ref{thm:bit} and~\ref{thm:main}.
%
\medskip
\begin{theorem}
\label{thm:bit} If $\Cp > \Numzo - \Numzi$ then Alice can communicate a single bit correctly to Bob (while keeping it secret from Calvin) using codes of
computational complexity $O(\poly(\Cp,\log_2 q))$ and error probability $O(q^{-C})$.
\end{theorem}
\medskip

Combining the codes in Theorem~\ref{thm:bit} with the ``shared-secret'' codes in~\cite{Jaggi++2008} then gives us the following theorem.

\medskip
\begin{theorem}
\label{thm:main} No rate higher than $\Cp - \Numzo - \Numzi$ is
achievable. A rate of $\Cp - \Numzo - \Numzi$ is achievable with
codes of computational complexity $O(n\poly(\Cp,\log_2 q))$.
\end{theorem}
\medskip

\noindent {\bf Note:}
In~\cite{Ngai.Yeung2009:Secure-Error-Correcting}, Ngai et al show
that $\Cp - 2\Numzo - \Numzi$ is an upper bound on the rate, assuming no
error events, and single-letter coding (respectively equations ($87$) and ($65$) in their proof).
Our work achieves higher rates by instead assuming asymptotically negligible probability of error, and block coding.

\subsection{High-level overview of proofs and techniques}

We first show in Section~\ref{sec:upperboud} that
$\Cp-\Numzo-\Numzi$ is an upper bound on the rate at which a secret
message can be correctly transmitted from Alice to Bob, by
demonstrating an attack that Calvin can use to successfully disrupt
communication if Alice tries to communicate at any higher rate.
We then construct efficient
codes that essentially achieve rate $\Cp-\Numzo-\Numzi$.
Our codes consist of the three layers
described below.  All the three layers are embedded along with
Alice's message into her packets and then transmitted through the
network using random linear network codes.

\noindent {\bf Secret-sharing layer:}
In Section~\ref{sec:single-bit} we {first prove Theorem~\ref{thm:bit} by showing} how to
communicate a single bit secretly and correctly over a network
containing adversaries that can jam and eavesdrop, as long as
$\Cp>\Numzi+\Numzo$. This layer is important for the error-control
layer described later, and can be implemented via a ``small" header
appended to each network coded packet. When $k$ secret bits are to
be shared, the scheme is repeated $k$ times in each transmitted
packet header, for a secret-sharing header of total length
$\Cp+k\Cp(\Cp-\Numzi)$. 
The secret-sharing layer consisting of the following components:

\noindent {\it $1$. Identity matrix:} As standard in random linear network
coding~\cite{Ho++2006}, \cite{Silva++2008}, the identity matrix
$I_C$ is appended to convey to the receiver information about the
linear transform induced by the random linear network code.

\noindent {\it $2$. Bit matrices:} For each secret
bit, $i \in \{1, \ldots, k\}$, if the $i$th secret bit equals $0$,
the $(\Cp-\Numzi)\times \Cp(\Cp-\Numzi)$ matrix $S^i$ (over $\Field_q$) is chosen as a zero matrix; otherwise, $S^i$
is chosen independently and uniformly at random from all
$(\Cp-\Numzi)\times \Cp(\Cp-\Numzi)$ matrices.
We refer to $S^i$ as a {\it bit matrix}.
The idea is that the
rank of the matrices corresponding to bit $0$ is much smaller than
the rank of the matrices corresponding to bit $1$---due to the
limitation on the numbers of packets Calvin can observe or inject,
with high probability he cannot change the rank of the corresponding
received matrix by too much. Details are given in
Lemma~\ref{Le:Trans-Bit}.

\noindent {\it $3$.Random matrix:}
Alice adapts the scheme of~\cite{Silva.Kschischang2008:ISIT} to keep the bit matrices secret from Calvin.
That is, for each secret bit $i$ that Alice wishes to communicate to Bob, she combines the bit matrix $S^i$ with a
random noise matrix $N^i$ (at rate $\Numzi$). It can be shown that
it is impossible for Calvin to glean any useful information (since
it can only eavesdrop at rate $\Numzi$).

Section~\ref{subsec:completerate} combines the secrecy layer with the two other layers described below to complete our code construction.

\noindent{\bf Secrecy layer:}
As done with the random matrices $N^i$ in the secret-sharing layer above, a random matrix $N$ is used to preserve the secrecy of the source message $S$ (of rate $\Cp-\Numzo-\Numzi$), yielding a encoded matrix $M$ (of rate $C-Z_O$).

\noindent {\bf Error control layer:}
In this layer Alice uses the ``shared-secret" scheme outlined in
Theorem $1$ of~\cite{Jaggi++2008}. That is, Alice first takes a
secret linear hash to her secrecy-encoded message $M$ 
to generate a small hash value. Both the linear hash and
the resulting hash value (say $k$ bits in all) are transmitted to
Bob using the secret-sharing layer. Alice then combines her data
with a zero-value matrix (of rate $\Numzo$), such that Bob can
use the secret hash to {\em distill} Alice's codeword $M$
from the corrupted information reaching the destination.

Vis-a-vis our secret-sharing scheme of
Section~\ref{sec:single-bit}, the work
of~\cite{Jaggi.Langberg2007} (by a subset of the authors of this
work) claimed to have the same result. However, we show in
Section~\ref{sec:errata} that the
scheme proposed in~\cite{Jaggi.Langberg2007} is incorrect by giving  an attack that Calvin can use
to ensure that Bob has a significant probability of decoding error.

%
%

\section{Network Model  and Problem Statement}\label{sec:model}
We use the general model proposed in~\cite{Jaggi++2008}. To simplify
notation we consider only the problem of communicating from a single
source to a single destination\footnote{Similarly to many network coding
algorithms, our techniques generalize to multicast problems.}.

\subsection{Network Model}
Alice communicates to Bob
over a network with an attacker (adversary) Calvin hidden somewhere in it. Calvin aims to
disrupt  the transfer of information from Alice to Bob and in the
meantime eavesdrop the information Alice sends. He can observe some
of the transmissions, and can inject his own fake transmissions.

Calvin is computationally unbounded,
knows the encoding and decoding schemes of Alice and Bob,
and the network code implemented
by the interior nodes.
He also knows the network topology,
and he gets to choose which network links to eavesdrop on and which ones to
corrupt.

The network is modeled as a directed and delay-free graph 
whose edges each have capacity equal to
one symbol of a finite field of size $\fs$,
$\Field_\fs$, per unit time\footnote{For ease of presentation edges
with non-unit capacities are not considered here (as in~\cite{Jaggi++2008}, they may be modeled via block coding
and parallel edges).}. All computations are
over $\Field_{\fs}$. The {\em network capacity},
denoted by $\Cp$, is the {\it min-cut from source to destination}\footnote{For the corresponding multicast case, $\Cp$ is defined as the
minimum of the min-cuts over all destinations. 
It is well-known that $\Cp$ also equals the time-average of the maximum
number of packets that can be delivered from Alice to Bob, assuming
no adversarial interference, i.e., the {\it max flow}.}.

Each packet contains $\Block$ symbols
from $\Field_{\fs}$. Alice's message is denoted $S \in \calS$. To
send this to Bob over the network, Alice encodes it into a matrix $X \in \Fq^{C \times \Bl}$, possibly using a \emph{stochastic encoder}\footnote{The random coin tosses
made by Alice as part of her encoding scheme are not known to either
Calvin or Bob.}. The $i^{th}$ row in
$\eX$ is Alice's $i^{th}$ packet. As in~\cite{Ho++2006}, Alice and internal nodes in take random linear combinations of their observed packets
to generate their transmitted packets.

Analogously to how Alice generates $\eX$, Bob
organizes received
packets into a matrix $\eY$.
The $i^{th}$ received packet
corresponds to the $i^{th}$ row of $\eY$. The random linear network
code used by Alice and all internal nodes induces a linear transform
$A$ from $\eX$ to $\eY$, such that $\eY=A\eX$ when no error is
induced by the adversary\footnote{For the ease of notation we assume Bob removes redundant incoming edges so that the number of edges reaching Bob equals the min-cut
capacity $\Cp$ from Alice to Bob.}. Thus $\eY$ is a matrix in
$\Field_q^{\Cp\times \Bl}$, and $A \in \Fq^{C \times C}$. Hereafter we assume that the matrix $A$ is invertible, which happens with high probability if $q$ is sufficiently large~\cite{Ho++2006}.

Calvin can eavesdrop on $\Numzi$ edges, and can inject (possibly fake) information at $\Numzo$ locations\footnote{We assume throughout that the information injected into the network by Calvin
is {\em added} to the original information transmitted
(here we consider addition over our field $\Field_{\fs}$).}, in the network. The matrix received by Bob is then  $Y = AX + Z$,
where $Z$ corresponds to the information injected by Calvin as seen by Bob. Note that the limitation of Calvin's jamming capacity implies that $\rank (Z) \leq Z_O$. Similarly, Calvin's observation can be described as a matrix $W = BX$,
where $B \in \Fq^{Z_I \times C}$ is the linear transform undertaken by $X$ as seen by Calvin.


\subsection{Problem Statement}
Alice wishes to communicate with Bob with perfect secrecy and vanishingly small error probability. That is, Alice's scheme is {\em perfectly secret} if
\begin{equation}\label{eq:secrecy-condition}
  I(S;W) = 0 \quad \forall B \in \Fq^{Z_I \times C}
\end{equation}
i.e., Calvin obtains no information about Alice's message.
The {\em error probability} is the probability
that Bob's reconstruction $\hat{S}$ of Alice's information $S$ is inaccurate, i.e., $P[\hat{S} \neq S]$. We consider the error probability of the worst-case scenario\footnote{Our interest is to design communication schemes that do not rely on the specific network topology or network code used.}. Namely, a scheme has error probability less than $\epsilon$ if
$  P[\hat{S} \neq S] < \epsilon \quad \forall A,Z$, where $A$ is assumed to be nonsingular, and $\rank (Z) \leq Z_O$.
The {\it rate} $\Rate$ of a scheme is the number of information bits of information Alice transmits to Bob, amortized by the size 
of a packet in bits, i.e., $R = \frac{1}{n} \log_q |\calS|$.
The rate $\Rate$ is said to be {\em achievable}
if for any $\epsilon > 0$, any $\delta>0$,
and sufficiently large $\Bl$,
there exists a perfectly secret block-length-$\Bl$ network code
with rate at least $R-\delta$ and a probability of error less than $\epsilon$.




\begin{table}[htb!]
\caption{Summary of commonly used notation} 
\centering 
\begin{tabular}{c c}
\hline\hline Notation&Meaning
\\ \hline $\Cp$& Capacity \\ \hline$\Numzi$ & Eavesdropping rate
\\ \hline $\Numzo$ & Jamming rate \\ \hline $\Bl$ & Packet length
\\ \hline $q$& Field size\\  \hline$Q = q^\Cp$ & Extension field size\\
\hline
\end{tabular}
\label{tab:SumResult}
\end{table}

\section{Converse for Theorem~\ref{thm:main}}
\label{sec:upperboud}

We start by presenting an attack that Calvin
may use to force the achievable rate to at most $\Cp - \Numzo -
\Numzi$, thereby demonstrating that this is indeed an upper bound on
the achievable rate. Let $\{e_1,e_2,...,e_\Cp\}$ be a set of edges
that form a cut from Alice to Bob. Calvin jams the edges in
$\{e_1,e_2,...,e_{\Numzo}\}$ by adding random errors on them.
Further, Calvin eavesdrops on edges in
$\{e_{\Numzo+1},e_{\Numzo+2},...,e_{\Numzo+\Numzi}\}$. Let ${\mathbf
X}$ be the random variable denoting Alice's information. Let
$\mathbf {Y}_j$, $\mathbf {Y}_e$, and $\mathbf {Y}_u$ be the random
variables denoting the packets carried by the jammed
edges $\{e_1,e_2,...,e_{\Numzo}\}$, eavesdropped edges
$\{e_{\Numzo+1},e_{\Numzo+2},...,e_{\Numzo+\Numzi}\}$, and untouched
edges $\{e_{\Numzo+\Numzi+1},e_{\Numzo+\Numzi+2},...,e_{\Cp}\}$
respectively. Let $\mathbf {Y}$ be the random variable denoting the
packets received by Bob. Then
\begin{eqnarray}
\Block\Rate & = & H(\mathbf{X}) =  H(\mathbf{X}|\mathbf{Y}) + I(\mathbf X; \mathbf {Y})\label{eq:bnd0}\\
& \leq & 1 + \epsilon\Block\Rate + I(\mathbf X; \mathbf {Y})\label{eq:bnd05}\\
&\leq& 1 + \epsilon\Block\Rate + I(\mathbf X; \mathbf {Y}_j, \mathbf {Y}_e, \mathbf {Y}_u)\label{eq:bnd1}\\
& = & 1 + \epsilon\Block\Rate + I(\mathbf X; \mathbf {Y}_e, \mathbf {Y}_u)\label{eq:bnd2}\\
& = & 1 + \epsilon\Block\Rate + I(\mathbf X; \mathbf {Y}_e)+I(\mathbf X; \mathbf {Y}_u | \mathbf {Y}_e)\label{eq:bnd3}\\
& = & 1 + \epsilon\Block\Rate + I(\mathbf X; \mathbf {Y}_u | \mathbf {Y}_e)\label{eq:bnd4}\\
&\leq& 1 + \epsilon\Block\Rate +  H(\mathbf {Y}_u)\label{eq:bnd5}\\
& \leq& \Block\left [(\Cp-\Numzi-\Numzo) + \epsilon\Rate +
\frac{1}{\Block}\right ].\label{eq:bnd6}
\end{eqnarray}
Here (\ref{eq:bnd0}) follows from the fact that Alice's message is
uniformly distributed over ${\mathbf X}$, (\ref{eq:bnd05}) from
Fano's inequality, (\ref{eq:bnd1}) from the data processing
inequality, (\ref{eq:bnd2}) since Calvin adds random noise on the
edges he jams  and so $\mathbf {Y}_j$ is independent of $(\mathbf X,
\mathbf {Y}_e, \mathbf {Y}_u)$, (\ref{eq:bnd3}) by the chain rule
for mutual information, (\ref{eq:bnd4}) from the fact that
information-theoretic secrecy is required and so $I(\mathbf X;
\mathbf {Y}_e)=0$, (\ref{eq:bnd5}) by the fact that conditioning
reduces entropy and the definition of mutual information, and
finally (\ref{eq:bnd6}) by the fact that there are at most
$\Cp-\Numzi-\Numzo$ links corresponding to the random variable
$\mathbf {Y}_u$ and the alphabet-size upper bound on entropy.
Requiring $\epsilon \rightarrow 0$ as $\Block \rightarrow \infty$
gives the required result.

\section{Auxiliary Tools}

\newcommand{\zi}{Z_I}
\newcommand{\zo}{Z_O}

\subsection{Secrecy Coding}
\label{ssec:secrecy-layer}

Consider a special case of the problem where Calvin can eavesdrop $\zi < C$ packets but cannot jam any packets ($\zo = 0$). Below, we review a construction of a perfectly secret scheme that asymptotically achieves the maximum possible rate (i.e., the secrecy capacity) $R=C-\zi$. The scheme, proposed in \cite{Silva.Kschischang2008:ISIT}, is based on MRD codes. (For more details on MRD codes, see~\cite{Silva.Kschischang2008:ISIT}.)

Let $Q = q^C$ and let $\FQ$ be an extension field of $\Fq$. Let $\phi: \FQ \to \Fq^{1 \times C}$ be a vector space isomorphism. In addition, let $\phi_{m,n}: \FQ^{m \times n} \to \Fq^{m \times Cn}$ be a vector space isomorphism such that the $i$th row of $\phi_{m,n}(X)$ is given by $\mat{\phi(X_{i,1}) & \cdots & \phi(X_{i,n})}$. In other words, we expand each element of $X \in \FQ^{m \times n}$ as a length-$C$ row vector over $\Fq$ (with the number of columns in matrix increasing accordingly). We will omit the subscript from $\phi_{m,n}$ when the dimensions of the argument are clear from the context.

Let $H \in \FQ^{(C-\zi) \times C}$ be the parity-check matrix of a $[C,\zi]$ linear MRD code over $\FQ$. Let $T \in \FQ^{C \times C}$ be an invertible matrix chosen such that the first $C-\zi$ rows of $T^{-1}$ are equal to $H$. Assume that $n$ is divisible by $C$ and let $n' = n/C-1$.

In order to encode a given message $S \in \FQ^{(C-\zi) \times n'}$, Alice first generates a random matrix $N \in \FQ^{\zi \times n'}$ uniformly and independently from any other variables. Then, she computes $X = \mat{I_C & \phi(x)}$, where $x = T \mat{S \\ N}$.

After receiving $Y = AX = \mat{A & A\phi(x)}$, Bob computes $X = A^{-1} Y$ to recover $x = \phi^{-1}(\phi(x))$. Then, Bob can easily obtain $S$ since, by construction, $S = H x$.

Recall that Calvin's observation is given by $W = BX$, where $B \in \Fq^{\zi \times C}$. According to Theorem~4 of \cite{Silva.Kschischang2008:ISIT}, we have that $I(S;W) = 0$ for all $B$, and therefore (\ref{eq:secrecy-condition}) is satisfied. Thus, the scheme is indeed perfectly secret. 

The decoding complexity is given by $O(nC^2)$ operations in $\FQ$, which can be done in $O(nC^4)$ operations in $\Fq$.

\subsection{Error Control under a Shared Secret Model}
\label{ssec:error-control-layer}

Consider now the case where Calvin can jam $\zo < C$ packets and eavesdrop \emph{any} number of packets he choose. However, we drop the requirement of secret communication, i.e., all we require is that Bob can decode correctly.
In addition, suppose the existence of a {\em low} rate side channel, which Calvin cannot access, that enables Alice to transmit to Bob a small secret $\mathbb{S}$. Below, we review a coding scheme presented in \cite{Jaggi++2008} that can asymptotically achieve the maximum possible rate $R = C-\zo$.

Let $b = C-\zo$. We first describe how Alice produces the secret bit string $\mathbb{S}$ based on a given message $M \in \Fq^{b \times (n-b)}$. To begin with, she generates $\alpha=bC+1$ symbols $\rho_1,\rho_2,...,\rho_\alpha \in \Fq$ independently and uniformly at random. Let $P \in \Fq^{n \times \alpha}$ be the matrix given by $P_{(i,j)} = (\rho_j)^i$. Then, she computes a matrix $\mathbb{H} = \bar{X}P \in \Fq^{b \times \alpha}$, where $\bar{X} = \mat{I_b & M}$. The tuple $(\rho_1,\rho_2,...,\rho_\alpha, \mathbb{H})$, consisting in total of $\alpha(b+1)$ symbols in $\Fq$, comprises the message ``hash'' that should be secretly transmitted to Bob. The bit representation of this tuple yields the string $\mathbb{S} \in \{0,1\}^k$, consisting of $k = \alpha(b+1) \log_2 q$ bits. Over the main channel, Alice transmits the $C \times n$ matrix $X = \mat{\bar{X} \\ 0} = \mat{I_b & M \\ 0 & 0}$.

Assuming that $(\rho_1,\rho_2,\ldots,\rho_\alpha, \mathbb{H})$ is secretly and correctly received by Bob, let us proceed to the description of Bob's decoder. First, Bob reconstructs the matrix $P$. Bob obtains $Y = A X + Z$, where $Z \in \Fq^{C \times n}$ has rank at most $Z_O$. This can also be written as $Y = \tilde{A} \bar{X} + Z$, where $\tilde{A}$ consists of the first $b$ columns of $A$. Let $\bar{Y}$ be the reduced row echelon form of $Y$. It is shown in \cite{Jaggi++2008} that, with probability at least 
$1 - O(1/q)$ for any fixed network,
$\bar{X}$ can be written as $\bar{X} = U\bar{Y}$ for some $U \in \Fq^{b \times C}$. It is also shown in \cite{Jaggi++2008} that, with probability at least $1 - n^\alpha/q$, the system $U\bar{Y}P = \mathbb{H}$ has a unique solution in $U$. Bob solves this system to find $U$, computes $\bar{X} = U \bar{Y}$ and finally recovers $M$.

Overall, the probability of error of the scheme is at most 
$n^\alpha/q + O(1/q) = O(n^{C^2}/q)$,
while the decoding complexity is $O(nC^3)$ operations in $\Fq$.

\section{Sending a Single Bit Secretly and Reliably}
\label{sec:single-bit}

Let $C' = C-Z_I$. In this section, we show how Alice can transmit a secret bit reliably to Bob when $C>\zi+\zo$. We assume that $n = C(1+C')$, as this is the smallest packet length required for the scheme to work. Larger packet lengths can be easily handled by zero-padding the transmitted packets.

Let $T \in \FQ^{C \times C}$ and $H \in \FQ^{C' \times C}$ be as given in Section~\ref{ssec:secrecy-layer}.

\subsection{Alice's encoder}

Initially, Alice chooses a matrix $S \in \FQ^{C' \times C'}$ according to her secret bit: if the bit is 1, she picks $S$ uniformly at random; otherwise, if the bit is 0, she sets $S = 0$. Then, she sends $S$ to Bob using the secrecy scheme described in Section~\ref{ssec:secrecy-layer}. More precisely, she transmits $X = \mat{I_C & \phi(x)}$, where $x = T \mat{S \\ N}$ and $N \in \FQ^{\zi \times C'}$ is a uniformly random matrix chosen independently from $S$.

\subsection{Bob's decoder}

Recall that Bob receives a matrix $Y = AX + Z$, where $A \in \Fq^{C \times C}$ is nonsingular and $Z \in \Fq^{C \times C(1+C')}$ has rank at most $\zo$. Let $\bar{Y}$ denote the reduced row echelon form of $Y$. Consider first the case where $\bar{Y} = \mat{I & \phi(r)}$, for some $r \in \FQ^{C \times C'}$. It is possible to show that $Hr = S + E$, where $E \in \FQ^{C' \times C'}$ is a matrix of rank at most $\zo$. As will be shown later, with high probability, $Hr$ is full-rank if and only if Alice's secret bit is 1. Thus, Bob can decode by computing the rank of $Hr$.

In general, however, $\bar{Y}$ may not have the form described above. Nevertheless, as shown in \cite{Silva++2008,Silva2009(PhD)}, it is possible to extract from $\bar{Y}$ some matrices $r \in \FQ^{C \times C'}$, $\hat{L} \in \Fq^{C \times \mu}$ and $\hat{V} \in \FQ^{\delta \times C'}$ such that
\begin{equation}\nonumber
  r = x + \hat{L} V^1 + L^2 \hat{V} + L^3 V^3
\end{equation}
for some $V^1 \in \FQ^{\mu \times C'}$, $L^2 \in \Fq^{C \times \delta}$, $L^3 \in \Fq^{C \times \epsilon}$ and $V^3 \in \FQ^{\epsilon \times C'}$. Moreover, it is shown in \cite{Silva2009(PhD)} that $\mu,\delta \leq \zo$ and
\begin{equation}\nonumber
  \epsilon \leq \zo - \max\{\mu,\delta\}.
\end{equation}
Note that $\epsilon < C' - \max\{\mu,\delta\}$, since $\zo < C'$.

In possession of $r$, $\hat{L}$ and $\hat{V}$, Bob is now ready to decode the secrecy layer that has been applied to $x$.

We have
\begin{align}
Hr
&= H x + H \hat{L} V^1 + H L^2 \hat{V} + H L^3 V^3 \nonumber \\
&= S + \hat{\Lambda} V^1 + \Lambda^2 \hat{V} + \Lambda^3 V^3 \label{eq:secret-sharing-decoded-noncoherent}
\end{align}
where $\hat{\Lambda} = H \hat{L}$, $\Lambda^2 = H L^2$ and $\Lambda^3 = H L^3$.
Note
that $\hat{\Lambda} \in \FQ^{C' \times \mu}$ and $\hat{V} \in \FQ^{\delta \times C'}$ are known.

Now, let $J \in \FQ^{(C'-\mu) \times C'}$ and $K \in \FQ^{C' \times (C'-\delta)}$ be full-rank matrices such that $J \hat{\Lambda} = 0$ and $\hat{V} K = 0$. Then Bob can further simplify (\ref{eq:secret-sharing-decoded-noncoherent}) by computing
\begin{equation}\nonumber
  J H r K = J S K + J \Lambda^3 V^3 K.
\end{equation}
Note that $\rank (J \Lambda^3 V^3 K) \leq \epsilon < C' - \max\{\mu,\delta\}$.

Thus, Bob performs the following test. If $JHrK$ is full-rank, then Bob concludes that bit $1$ was sent; otherwise, Bob concludes that bit $0$ was sent.

With respect to complexity, computing $\bar{Y}$ takes $O(C^2n) = O(C^4)$ operations in $\Fq$. Computing $J$, $K$, $JHrK$ and the rank of $JHrK$ each take $O(C^3)$ operations in $\FQ$, which amounts to $O(C^5)$ in $\Fq$. Thus, the overall decoding complexity is $O(C^5)$ operations in $\Fq$.

\subsection{Probability of error analysis}

When bit 0 is sent, Bob never makes an error; he makes an error if and only if bit 1 is sent and $JHrK$ is not full-rank. Recall that, when bit 1 is sent, $S$ is uniformly distributed over $\FQ^{C' \times C'}$. Due to the secrecy encoding, Calvin has no information about $S$, and therefore $S$ is statistically independent from $\Lambda^3 V^3$. It follows that $S'=S+\Lambda^3 V^3$ is also uniformly distributed over $\FQ^{C' \times C'}$. Thus, the probability of error when bit 1 is sent is equal to the probability that $JS'K \in \FQ^{(C'-\mu) \times (C'-\delta)}$ is not full-rank for a uniform $S'$.

\medskip
\begin{lemma}\label{Le:Trans-Bit}
If $S' \in \FQ^{C' \times C'}$ is uniformly distributed then, for any $J \in \FQ^{(C'-\mu) \times C'}$ and any $K \in \FQ^{C' \times (C' - \delta)}$, the matrix $JS'K$ is full-rank with probability at least $1-C'/Q$.
\end{lemma}
\begin{proof}
  Without loss of generality, assume $\mu \geq \delta$. It suffices to prove the statement for $\mu=\delta$; if $\mu > \delta$, then removing $\mu-\delta$ columns from $K$ cannot possibly increase the rank of $JS'K$.


  For any fixed $J$ and $K$, consider the entries of $S'$ as variables taking values in $\FQ$. Then each entry of $JS'K$ is a multivariate polynomial over $\FQ$ with degree at most 1. It follows that $\det(JS'K)$ is a multivariate polynomial over $\FQ$ with degree at most $C'-\mu \leq C'$. Note that, if $Q \leq C'$, the statement follows trivially, so assume $Q > C'$. From \cite[Lemma 4]{Ho++2006}, we have that $P[\det(JS'K)=0] \leq C'/Q$.
\end{proof}
\medskip

Thus, the probability of error of the scheme is upper bounded by $C'/Q \leq C/q^C$, which can be made arbitrarily small by choosing $q$ sufficiently large. This proves Theorem~\ref{thm:bit}.

\section{Achievability for Theorem~\ref{thm:main}}
\label{subsec:completerate}

We now describe a coding scheme that achieves rate $R = C - Z_I - Z_O$ asymptotically in the packet length $n$.

As before, assume that $n$ is divisible by $C$ and let $n' = n/C - (1+kC')$, where $k = (bC+1)(b+1)\log_2 q$.

Let $H \in \FQ^{C' \times C}$ be the parity-check matrix of a $[C,\zi]$ linear MRD code over $\FQ$. Let $T \in \FQ^{C \times C}$ be an invertible matrix such that the first $C-\zi$ rows of $T^{-1}$ are equal to $H$.

Similarly, let $H_0 \in \FQ^{R \times b}$ be the parity-check matrix of a $[b,\zi]$ linear MRD code over $\FQ$, and let $T_0 \in \FQ^{b \times b}$ be an invertible matrix such that the first $R$ rows of $T_0^{-1}$ are equal to $H_0$.

\subsection{Alice's encoder}

First, given a message $S \in \FQ^{R \times n'}$, Alice computes $x = T_0 \mat{S \\ N}$, where $N \in \FQ^{Z_I \times n'}$ is chosen independently and uniformly at random. Then, she sets $M = \phi(x)$ and generates a string $\mathbb{S} \in \{0,1\}^k$ of $k$ bits according to the scheme described in Section~\ref{ssec:error-control-layer}. Next, for each $i$th bit of $\mathbb{S}$, Alice produces a matrix $S^i \in \FQ^{C' \times C'}$ according to the scheme described in Section~\ref{sec:single-bit}. Then, for each $i=1,\ldots,k$, she computes $x^i = T\mat{S^i \\ N^i}$, where each $N^i \in \FQ^{Z_I \times C'}$ is chosen uniformly at random and independently from any other variables. Finally, she produces a transmission matrix 
\begin{equation}\nonumber
  X = \mat{I_C & \phi(x^1) & \phi(x^2) & \cdots & \phi(x^k) & \mat{M \\ 0}}.
\end{equation}

\subsection{Bob's decoder}

For each $i=1,\ldots,k$, Bob extracts a submatrix $Y^i$ from $Y$ corresponding to the submatrix $\mat{I_C & \phi(x^i)}$ from $X$ (i.e., columns $1,\ldots,C,C+(i-1)C'+1,\ldots,C+iC'$). He then applies on $Y^i$ the decoder described in Section~\ref{sec:single-bit} to obtain each $i$th bit of $\mathbb{S}$. 

Similarly, Bob extracts a submatrix $Y^0$ consisting of the first $b$ and the last $n'C$ rows of $Y$. Note that $Y^0 = AX^0 + Z^0$, where $X^0 = \mat{I_b & M \\ 0 & 0} \in \Fq^{C \times (b+n'C)}$ and $Z^0$ has rank at most $\zo$. Then, Bob applies the decoder described in Section~\ref{ssec:error-control-layer} to obtain $M$.

Finally, Bob computes $x = \phi^{-1}(M)$ and $S = H_0x$.

\subsection{Overall Analysis}

\subsubsection{Secrecy analysis}
The secrecy of the message is guaranteed by the scheme of Section~\ref{ssec:secrecy-layer}.

\subsubsection{Error probability analysis} 
By the union bound, the probability that Bob makes an error when decoding the $k$-bit secret $\mathbb{S}$ is at most $k C/q^C \leq C^4 (\log_2 q)/q^C = O(\frac{\log_2 q}{q^C})$.
Given that the secret is decoded correctly, the probability that Bob makes an error when decoding the message is at most $O(n^{C^2}/q)$. Thus, the overall probability of error is at most $O(n^{C^2}/q)$.

\subsubsection{Rate analysis}
The rate of the scheme is given by $Rn'C/n = R(1- (1+kC')C/n) \leq R - R C^5(\log_2 q)/n$. Thus, the rate loss is $O(\frac{\log_2 q}{n})$.

\subsubsection{Complexity analysis}
Decoding all the secret bits takes $O(kC^5) = O(C^8 \log_2 q)$ operations in $\Fq$, while decoding the message is dominated by the secrecy decoding step with $O(C^4 n)$ operations in $\Fq$.

\medskip
{\it Note:} Both the rate loss and the error probability can be made asymptotically small by choosing $q$ to grow faster than polynomially but slower than exponentially in $n$. For instance, we may choose $q =2^{\lfloor \sqrt{n} \rfloor}$.

\section{Errata for~\cite{Jaggi.Langberg2007}} \label{sec:errata}
We briefly reprise the scheme of~\cite{Jaggi.Langberg2007} before demonstrating the flaw in the proof. In what follows, all operations are over $\Field_q$.

In the scheme of~\cite{Jaggi.Langberg2007} there exist two
hash matrices $D_0$ and $D_1$ which are chosen independently and
uniformly at random $\Cp^2(\Cp-\Numzo) \times \Cp^2$ {\it Vandermonde matrices}, {\it i.e.}, each column
of $D_0$ and $D_1$ is of the form $\mathbf
{h}(u)=[u,u^2,...,u^{\Cp^2(\Cp-\Numzo)}]^T$, where the generator $u$
is chosen independently and uniformly at random from $\Field_q$.
Both $D_0$ and $D_1$ are publicly known to all parties, including Bob and Calvin.

{\bf Alice's Encoder}: Alice first chooses a random length-$(\Cp^2(\Cp-\Numzo)-\Cp^2)$ row vector $\mathbf
u$. Let $I\in
\{0,1\}$ be the secret bit that Alice wishes to send to Bob. Alice then constructs the length-$1\times \Cp^2$ row vector $\mathbf r$ such that $[\mathbf u, \mathbf r]D_I=0$.
Note that such $\mathbf r$ exists since the last $\Cp^2$ rows of
$D_I$ form an invertible matrix. Finally the vector $[\mathbf
u,\mathbf r]$ is rearranged into a $(\Cp-\Numzo)\times \Cp^2$
matrix which is sent through the network via random linear network
coding.

{\bf Bob's Decoder}: After receiving the ${\Cp\times
\Cp^2}$ matrix $Y$, for each $I\in \{0,1\}$ Bob check whether there exists
$\Cp-\Numzo$ length-$\Cp$ vectors $\{ \mathbf {x}_i, i\in[1,{\Cp-\Numzo}]\}$ such that $[\mathbf {x}_1 Y, \mathbf {x}_2
Y,...,\mathbf {x}_{\Cp-\Numzo}Y]D_I=0$. If so, Bob decodes the secret bit as $I$. The idea is that if $I$ is Alice's bit, such $\{ \mathbf
{x}_i, i\in[1,{\Cp-\Numzo}]\}$ exists for $D_I$ with high
probability~\cite{Jaggi++2008}.

{\bf Calvin's successful attack}: When Calvin corrupts $\Numzo\geq\Cp-\Numzo$
edges, Calvin could mimic Alice's behaviour when she wishes to transmit a particular bit, say
$1$. As a result Bob would always find length-$\Cp$ row vectors $\{ \mathbf {x}_i,
i\in[1,{\Cp-\Numzo}]\}$ such that
$[\mathbf {x}_1 Y, \mathbf {x}_2 Y,...,\mathbf
{x}_{\Cp-\Numzo}Y]D_1=0$. In this case Bob cannot determine
whether the bit $1$ is from Alice or from Calvin.

Even if Calvin can only inject $\Numzo<\Cp-\Numzo$ errors, if
$\Numzo+\Numzi\geq \Cp-\Numzo$, there is another successful attack for Calvin.
To see that, without loss of generality let
$\Numzo+\Numzi=\Cp-\Numzo$. Since Calvin can eavesdrop on $\Numzi$
packets $\{\mathbf {y}_i,i\in [1,\Numzi]\}$, he can carefully
choose his $\Numzo$ injected error packets $\{\mathbf {z}_i,i\in
[1,\Numzo]\}$ so that $[\mathbf {y}_1,...,\mathbf
{y}_{\Numzi},\mathbf {z}_1,...,\mathbf {z}_{\Numzo}]D_1=0$. In this
case, Bob also always decodes its bit as $1$. Thus
the scheme in~\cite{Jaggi.Langberg2007} only works for the case
where $\Cp>2\Numzo+\Numzi$, which does not improve the result
in~\cite{Jaggi++2008}.

{\bf Why our scheme works}: In our scheme
Section~\ref{sec:single-bit}, instead of distinguishing the bit
by the hash matrices, Alice hides her secret in the rank of the
bit matrix she transmits. In particular, there is a rank gap
$\Cp-\Numzi$ between the bit matrix for bit $0$ and the one for
bit $1$. Thus as long as $\Cp-\Numzi>\Numzo$, Calvin cannot mimic
Alice any more, since he can only inject $\Numzo$ errors. As a result Bob can
determine Alice's bit by examining the rank of the matrix he decodes.

\section{Conclusion}

In this work we considered the problem of communicating information secretly and reliably over a network containing a malicious eavesdropping and jamming adversary. Under the assumptions that vanishingly small probabilities of error and block coding are allowed, we substantially improve on the best achievable rates in prior work~\cite{Ngai.Yeung2009:Secure-Error-Correcting}, and also prove the optimality of our achievable rates. A key component of our code design is a scheme that allows a small amount of information to be transmitted secretly and reliably over the network, as long as the total number of packets that the adversary can either eavesdrop on or jam is less than the communication capacity of the network. In proving this scheme we correct an error in the proof of prior work~\cite{Jaggi.Langberg2007} by a subset of the authors of this work.

\bibliographystyle{IEEEtran}
\bibliography{IEEEabrv,jammeave}

\end{document}